\documentclass{article}

\usepackage{setspace}

\usepackage{amsmath}
\usepackage{amsfonts}
\usepackage{amssymb}
\usepackage{amsthm}
\usepackage{graphics}
\usepackage{enumerate}
\usepackage{color}
\usepackage{tikz-cd}
\usepackage{cleveref}
\usepackage{tensor}
\usepackage{booktabs}
\usepackage{apacite}
\usepackage[normalem]{ulem}

\newcommand{\R}{\mathbb{R}}
\newcommand{\bs}{\boldsymbol}
\newtheorem{theorem}{Theorem}[section]
\newtheorem{corollary}[theorem]{Corollary}

\newtheorem{lemma}[theorem]{Lemma}
\newtheorem{remark}[theorem]{Remark}
\theoremstyle{definition}
\newtheorem{definition}[theorem]{Definition}

\DeclareMathOperator{\ER}{ER}
\DeclareMathOperator{\RR}{RR}

\DeclareMathOperator{\ob}{ob}
\DeclareMathOperator{\argmin}{argmin}
\DeclareMathOperator{\Aut}{Aut}
\DeclareMathOperator{\Set}{Set}

\DeclareMathOperator{\fun}{fun}
\newcommand{\ed}{\mathrm d}

\newcommand{\red}[1]{{#1}}

\begin{document}

	\title{The Markowitz Category}
	\author{John Armstrong \\ King's College London}
	\date{}
	
	\maketitle

\begin{abstract}
We give an algebraic definition of a Markowitz market and classify
markets up to isomorphism. Given this classification,
the theory of portfolio optimization in Markowitz markets without short selling constraints becomes trivial.
Conversely, this classification shows that, up to isomorphism, there is
little that can be said about a Markowitz market that is not already detected
by the theory of portfolio optimization. In particular, if one seeks to develop a simplified low-dimensional model of a large financial market using mean--variance analysis alone, the resulting model can be at most two-dimensional.
\end{abstract}

\section*{Introduction}

When developing financial models there is a tension between the desire to capture the complexity of financial markets, and the need to simplify, both for for tractability and to avoid over-fitting. This leads one to consider the question of how best to produce low-dimensional approximations to high-dimensional financial models.
As an example of such a dimensional reduction, consider the celebrated one and two mutual fund theorems \cite{merton}. These build on the work of Markowitz in \cite{markowitz} and tell us that, in the Markowitz market model with no restrictions on short selling, an investor who is only interested in the optimal investment problems can safely ignore all but a two-dimensional subspace of the space of portfolios.

This paper considers what can happen if one's interests are more broad-ranging than just the classical optimal investment problem of Markowitz. Are there other low-dimensional subspaces of the Markowitz market model that may be of particular interest to other market players? We will prove that, in a clearly defined sense, the answer to this question is no. Moreover, in the same clearly defined sense, the two mutual fund theorem says all that there is to say about the market. The key, of course, is to give a rigorous explanation of what we mean by this ``clearly defined sense''. This is where we use a little category theory.

Category theory, introduced in \cite{eilenbergmaclane}, formalises the common practice of mathematicians to investigate categories of object up to some notion of equivalence or isomorphism. For example, one might attempt to classify
vector spaces up to bijective linear transformation or finite groups up to
group isomorphism. The advantage of this approach is that spurious details
are ignored. For example the specific set underlying the vector space or the group are irrelevant to their classification up to isomorphism.

Following a similar pattern, in \Cref{section:category} we will define a class of
objects called Markowitz markets and define a notion of a Markowitz isomorphism
between markets. Briefly, a Mark\-o\-witz market is a vector space of possible investment portfolios equipped with: a linear functional that gives the cost of each portfolio; a linear function giving the expected payoff of each portfolio; and a symmetric bilinear form that measures the covariance of two portfolios. An isomoprhism is a map that preserves these structures.

By defining the notion of isomorphism we formally define what we consider to be a financially meaningful feature of a Markowitz market, and what we consider to be spurious information. A financially meaningful property should be preserved by isomorphisms. For example the name of a specific
stock is not financially meaningful and our notion of isomorphism reflects this.

We note that this notion of isomorphism presupposes that risk can be measured adequately by standard deviation. As is well known, there are good reasons for considering other risk-measures, in which case one would require more data to define the market and one would have a different notion of isomorphism. Note, while our theory is predicated on the use of standard deviation to measure risk, it is not dependent upon the distribution of returns. In particular the normal distribution will not play a role in our theory.

Having identified the notion of isomorphism, we then classify all arbitrage-free Markowitz markets up to Markowitz isomorphism in \Cref{theorem:classification}. This is the central result of this paper. The proof only requires elementary linear algebra and can be given without considering portfolio optimization at all. 

In \Cref{section:optimization} we will show how our classification of
Markowitz markets can be applied to the study of portfolio optimization. We will
see that classical results such as the mutual fund theorems are immediately
obvious corollaries of our classification. Moreover, we will observe a close relation between risk-return diagrams and the classification of markets. For
example, we will see that two markets of the same dimension and containing no spurious portfolios of zero cost, zero risk and zero expected payoff are isomorphic if and only if they have the same efficient frontier.

As we shall see, the category theory approach to the problem is in many ways more general and more illuminating than the classical approach of
\cite{merton}. The classical approach is based on direct
calculation and the theory of Lagrange multipliers, while we 
geometric arguments based on the Gram--Schmidt process. \red{Readers who wish to compare
our presentation with the more standard presentations in terms of returns and portfolio weights should consult Appendix \ref{appendix} where we describe in detail there how to translate between the two approaches and give a numerical example}.

We will show how our geometric approach can often be generalized to situations where invariance under Markowitz isomorphisms is broken by choosing another
appropriate category. For example, when
considering the performance of an individual portfolio relative to the market
one should only consider isomorphisms that preserve this portfolio. To use the
jargon of category theory jargon, one is interested in the ``pointed category'' of Markowitz markets with a marked portfolio of cost $1$. This category is classified in \Cref{theorem:pointedCategory}. This theorem explains why risk-return diagrams such as \Cref{fig:efficientfrontier} are such an effective tool for understanding this problem. It is interesting to note that when considering optimal hedging, as is done in \cite{sharpe}, one again seeks a classification of markets with a marked portfolio (this time the asset to be hedged defines the marked portfolio). A priori, one might imagine that analysing the performance of a portfolio is a very different problem from the analysis of hedging a portfolio, yet both problems can be understood using the same classification theorem.

Another generalization we consider is a market with two marked portfolios. This problem naturally occurs in the Capital Asset Pricing Model (CAPM) (as described in, for example,  \cite{jensen} and originally developed in \cite{treynor,sharpeCAPM,lintner,mossin}) . We will show how this theory can be understood via an appropriate classification theorem, in this case \Cref{theorem:capm}. The approach can
be generalized further to include many classical generalizations of CAPM or to derive new results. For example, if one wishes to study the performance of different hedging portfolios using mean-variance analysis one is naturally lead to the question of classifying markets with yet more marked portfolios. Hence it would be straightforward to generalize the CAPM to obtain a model for evaluating the relative performance of hedging portfolios.

In \Cref{section:dimensionReduction} we show that our approach can be used to derive new financially significant results. We formally state and prove a mathematical version of our claim that there are no low-dimensional subspaces of a high dimensional market model that are of special interest to particular market players other than those given by the two mutual fund theorem. Our essential assumption in proving this result is that market players are only interested in markets up to Markowitz isomorphism. Our claim will then follow
from our classification theorem together with some very general ideas derived from category theory, which we summarize in Section \ref{section:invariantDefinitions}.

As a concrete and financially relevant example, consider the practice of applying principal component analysis to the correlation matrix in order to identify interesting subspaces of a market model. This allows one to identify higher dimensional subspaces of a financial model, but at the expense of breaking invariance up to
Markowitz isomorphism. Principal component analysis of the correlation matrix can be justified if one believes that the financial properties of a single stock
and of a basket of stocks are fundamentally different. For example, if one seeks to find specific stocks reflect the market as accurately as possible, Markowitz invariance is broken and principal component analysis may be a useful tool. On the other hand, if one seeks to choose a small number of individual stocks that represent the market as accurately as possible, one cannot go beyond the two mutual fund theorem.

\red{We give two further examples of how our
result can be applied in \Cref{section:dimensionReduction}. Specifically in Section \ref{section:uncertainty} we consider
the important problem of estimating the expected return using historic data and the resultant model uncertainty. This problem has been studied extensively (see for example, \cite{blackLitterman}, \cite{garlappiUppalWang}, \cite{jorion}, \cite{ceriaStubbs}). In Section \ref{section:hedging} we then consider the problem of designing a mutual fund to attract investors with existing liabilities. In the case where the potential investor's liability is known this problem has been studied before in \cite{sharpe}, but we will consider the case where the potential investor's liability is unknown. For both the problem of model uncertainty and the problem of investor's with existing, but unknown, liabilities, our result shows that
one cannot identify interesting portfolios beyond those identified by the two mutual fund theorem without supplying additional data.
}

Finally, we note that although we have chosen to phrase our results in 
terms of financial markets, we observe in \Cref{remark:linearSDE} that our
results also yield a classification for linear stochastic differential equations.
Thus one should expect theorems analogous to the two mutual fund theorem to be ubiquitous in the study of linear stochastic differential equations and hence in the study of the short time behaviour of stochastic differential equations in general.

\section{The Markowitz Category}
\label{section:category}

We begin with a formal definition of our category of markets. We will then describe how these markets arise in finance. We then prove a classification theorem for these markets.

\begin{definition}
A {\em Markowitz market} $(V,r,c,p)$ consists of a finite dimensional real vector space $V$ together with the data:
\begin{enumerate}[(i)]
\item  A symmetric bilinear map $r:V \times V \to \R$ satisfying $r(v,v)\geq0$
for all $v \in V$;
\item  Two linear functionals $c:V \to \R$ and $p:V \to \R$.
\end{enumerate}
\end{definition}

\begin{definition}
A {\em Markowitz morphism} between two Markowitz markets $(V,r,c,p)$ and
$(V^\prime, r^\prime, c^\prime, p^\prime)$ is a linear transformation
$T:V \to V^\prime$ which satisfies:
\begin{equation}
r^\prime(T v, T v) = r( v, v ) \qquad \forall v \in V,
\end{equation}
\begin{equation}
p^\prime(T v) = p(v) \qquad \forall v \in V,
\end{equation}
\begin{equation}
c^\prime(T v) = c(v) \qquad \forall v \in V.
\end{equation}
Two Markowitz markets are said to be isomorphic if there is a bijective
Markowitz morphism from one to the other.
\end{definition}

\red{Together our definition of markets and their morphisms defines what is called a {\em category}. Other examples of categories include: vector spaces and their linear transformations; topological spaces and their continuous maps; groups and their homomorphisms. We will review some essential definitions from category theory, including the definition of a category, in Section \ref{section:invariantDefinitions}. Until then we will not need to use any category theory explicitly.}

Markowitz markets naturally arise in finance.

Consider a trader who
buys and sells $n$ financial assets. The trader is interested in studying
portfolios made up from these assets. A portfolio is defined by knowing the
vector in $\R^n$ that contains the quantity of each asset held. The abstract vector space $V$ in our definition of a Markowitz market represents the space of possible portfolios. A portfolio may contain a negative quantity of a particular asset, this is interpreted financially by saying that a trader may choose to buy assets (a positive quantity) or borrow them (a negative quantity).

In this financial setting, the linear functional $c$ computes the initial cost of setting up a portfolio. If we assume the market is infinitely liquid and that unlimited amounts of each asset can be bought and sold it is reasonable to assume that the cost is indeed linear.

The trader models the financial assets as random variables. The linear functional $p$ computes the expected payoff of the portfolio at some future time $T$. Infinite liquidity and infinite market depth justify the assumption that $p$ is linear. The symmetric bilinear map $r$ computes the covariance of the two portfolios at the future time $T$. Note that here we are assuming that all the assets have finite variance.

The quantity $\sqrt{r(v,v)}$,
(the standard deviation of $v$), should be thought of as the {\em risk} of a portfolio $v$. There is an extensive literature on risk measurement and numerous statistical quantities have been proposed that can be used to measure
the risk of a portfolio. We will not debate the pros and cons of different risk measures here, we simply state that, in the Markowitz framework, risk is measured using standard deviation.

To justify the definition of a Markowitz morphism we assume that the trader
is only interested in the portfolios that are available, their costs, payoffs and risk measured using the standard deviation. The trader sees all other market data as extraneous. In particular the trader is unconcerned by the question of how many assets are combined to produce a portfolio. 

Our aim now is to classify Markowitz markets up to isomorphism. This is an elementary exercise in linear algebra. To reduce
the number of cases in our classification, we will only classify arbitrage-free
markets. These are defined as follows.

\begin{definition}
A {\em \red{Markowitz} arbitrage portfolio} is a portfolio $v \in V$ satisfying $r(v,v)=0$, $c(v)=0$ and $p(v)>0$. A Markowitz market is {\em arbitrage-free} if it
does not contain any \red{Markowitz} arbitrage portfolios.
\end{definition}
\red{If we were to choose a probability model for the asset payoffs compatible with $p$ and $r$ then we would define a classical arbitrage to be a portfolio of zero cost which has an almost surely non-negative payoff and a positive probability of a positive payoff. A Markowitz arbitrage is always a classical arbitrage, but the converse does not hold. Given any values $P$ for the expected payoff and $R$ for the variance we can always find a probability distribution with mean $P$ and variance $R$ which takes positive and negative values with positive probabilities (for example a normal distribution). Hence a Markowitz market is arbitrage-free as defined above if and only if it contains no classical arbitrages whatever compatible probability model is chosen for the payoff distribution. This justifies the use of the term arbitrage-free in our definition of an arbitrage-free Markowitz market.}

\begin{definition}
A portfolio $v \in V$ is said to be {\em risk-free} if $r(v,v)=0$.
A portfolio $v \in V$ is said to be {\em costless} if $c(v)=0$.
A portfolio $v \in V$ is said to be {\em valueless} if $r(v,v)=0$, $c(v)=0$
and $p(v)=0$.
\end{definition}

\begin{lemma}
If $T$ is a Markowitz morphism between $(V,r,c,p)$ and ($V^\prime, r^\prime, c^\prime, p^\prime)$ then
\[
r^\prime(T v_1, T v_2) = r(v_1, v_2) \qquad \forall v_1, v_2 \in V.
\]
\label{lemma:polarizationIdentity}
\end{lemma}
\begin{proof}
This follows immediately from the polarization identity for symmetric bilinear maps:
\begin{equation}
\label{polarizationIdentity}
r(v_1,v_2)=\frac{1}{4}\left(r(v_1+v_2,v_1+v_2) - r(v_1-v_2,v_1-v_2)\right).
\end{equation}
This shows that the entire covariance structure $r$ can be deduced from
knowing the standard deviation $r(v,v)$.
\end{proof}

\begin{lemma}
Define the linear map $\tilde{r}:V \to V^*$ by $\tilde{r}(v)(w)=r(v,w)$ then
the set of risk-free portfolios, $V^0$, is equal to $\ker \tilde{r}$.
\end{lemma}
\begin{proof}
If $v \in \ker \tilde{r}$ then $r(v,v)=\tilde{r}(v)(v)=0$. So  $\ker \tilde{r} \subseteq V^0$.

On the other hand, if $r(v,v)=0$ then the function $n(v)=r(v,v)$ has a local minimum at $v$. So the derivative of $n$ in any direction $w \in V$ is equal to zero. This derivative is equal to $2r(v,w)=2\tilde r(v)(w)$.  So $V^0 \subseteq \ker \tilde{r}$.
\end{proof}
\begin{corollary}
If we have a decomposition $V=V^0 \oplus V^1$ for some vector subspace $V^1$
then the value of $r$ on $V$ is determined by its value on $V^1$.
\label{lemma:directSumDecomposition}
\end{corollary}
\begin{proof}
Let $v=v_0+v_1$ where $v_0 \in V^0$ and $v_1 \in V^1$. Then
\begin{equation*}
\begin{split}
r(v,v)&=r(v_0,v_0)+2r(v_0,v_1)+r(v_1,v_1) \\
&= r(v_1,v_1)
\end{split}
\end{equation*}
The result now follows from \Cref{lemma:polarizationIdentity}.
\end{proof}

If a portfolio satisfies $r(v)=0$, $c(v)=0$ and $p(v)\neq0$ then
either $v$ or $-v$ will be a \red{Markowitz} arbitrage portfolio. So a Markowitz market is arbitrage-free if and only if all costless, risk-free portfolios are
valueless. This yields the following result:
\begin{lemma}
[Classification of arbitrage-free riskless markets]
In an arbitrage-free Markowitz market, we can write $V^0 = V^R \oplus ((\ker c) \cap V^0 )$ where $V^R$ is zero or one dimensional. If $V^R$ is
one dimensional it is spanned by a single portfolio $v_R$ of cost $1$.
$p=0$ on $(\ker c) \cap V^0$.
\label{lemma:risklessMarkets}
\end{lemma}

We are now ready to state and prove our main \red{mathematical} result which is to give a canonical form for all arbitrage-free Markowitz markets.

The canonical forms
will be expressed in terms of of the vector space $\R^n$. We will write
the bilinear map $r$ on $\R^n$ as an $n \times n$ matrix ${\bf r}$
such that
\[ r(v,w)=v^T {\bf r} w. \]
We will write the linear functionals $c$ and $p$ as co-vectors. We will write
the matrices $\bf{r}$ in block diagonal form and will use the notation $1_k$ for
the $k \times k$ identity matrix and will use $0$ for matrices of zeros whose dimensions can be deduced from the context.

\begin{theorem}
\label{theorem:classification}
We have the following classification of Markowitz markets.
\begin{enumerate}[(a)]
\item \underline{The case $c\neq0$}.

Let $n$ be given. Given four parameters $(k,m,g,i) \in \{0,1, \ldots n\}
\times [0, \infty) \times [0, \infty) \times \R$ which do not lie in 
the set
\begin{equation}
E_n = \{ (k,m,g,i) : (k=n \hbox{ and }m=0) \hbox{ or } (k=0 \hbox{ and }m \neq 0 ) \}
\label{enDefinition}
\end{equation}
we can define an isomorphism class of Markowitz markets, ${\cal M}^n_{k,m,s,q}$, as follows:
\begin{enumerate}[(i)]
\item If $m=0$, ${\cal M}^n_{k,m,g,i}$ is the isomorphism class of the market
$\R^n$ with
\[
{\bf r} = \left( \begin{array}{cc}
1_k & 0  \\
0 & 0
\end{array} \right), \quad c = (0,0,\ldots,0,1), \quad p=(g,0,\ldots,0,i).
\]
\item If $m \in (0, \infty)$, ${\cal M}^n_{k,m,g,i}$ is the isomorphism class of the market
$\R^n$ with
\[
{\bf r} = \left(\begin{array}{cc}
1_k & 0 \\
0 & 0
\end{array} \right), \quad c = \left(\frac{1}{m},0,\ldots,0\right), \quad p=
\begin{cases}
(\frac{i}{m},0, \ldots, 0 ) & \text{if } k =1 \\
(\frac{i}{m},g,0,\ldots,0) & \text{otherwise}.
\end{cases}
\]
Note that when $k=1$ the parameter $g$ is ignored. We have chosen our coordinates $m$ and $i$ for the isomorphism classes so that these variables will have simple geometric and financial explanations. This justifies the apparently unnecessary complexity of using $\frac{1}{m}$ and $\frac{i}{m}$ in the formulae.
\end{enumerate}
Any arbitrage-free Markowitz market of dimension $n$ with $c \neq 0$ belongs to one of these isomorphism classes. The isomorphism classes ${\cal M}^n_{k,m,g,i}$  are distinct except that
\begin{equation}
\label{equivalence1}
\hbox{ if } m \in (0, \infty) \hbox{, then }{\cal M}^n_{1,m,g,i} = {\cal M}^n_{1,m,g^\prime,i} \quad \forall\, g, g^\prime .
\end{equation}
\item \underline{The case $c=0$}.

Any arbitrage-free Markowitz market of dimension $n$ with $c$ identically
zero is Markowitz isomorphic to the market $\R^n$ with
\[
{\bf r} = \left( \begin{array}{cc}
1_k & 0  \\
0 & 0
\end{array} \right), \quad c=(0,0,\ldots,0), \quad p=(g,0,\ldots,0)
\]
where $k$ is a uniquely determined integer between $0$ and $n$. $g=0$
if $k=0$ but otherwise, $g$ is a uniquely determined element of $[0,\infty)$.
\end{enumerate}
\end{theorem}
\begin{proof}
We first assume that $c\neq0$. Case (i) and (ii) can be distinguished in an invariant fashion since there is a risk-free portfolio $v_R$ with $c(v_R)\neq 0$ in case
(i) but not in case (ii). Let us show that conversely if there is such a portfolio we can find a basis such that the market takes the form of case (i), 
and if not, it takes the form in case (ii).
\begin{enumerate}[(i)]
\item  We suppose that a risk-free portfolio with non-zero cost, $v_R$, exists. Take $e_n=\frac{v_R}{c(v_R)}$ and take $k=n-\dim V^0$. Take $\{
e_{k+1}, \ldots e_{n-1}\}$ to be a basis for $(\ker c) \cap V^0$. By \Cref{lemma:risklessMarkets}, $p$ is equal to $0$ on $(\ker c) \cap V^0$. Extend
$\{
e_{k+1}, \ldots e_{n-1}\}$ to a basis $\{v_1, \ldots, v_k,e_{k+1},\ldots, e_{n-1}\}$ for $\ker c$. Let $V_k$
be the span of $\{v_1, \ldots v_k\}$. Then $r$ restricted to $V_k$ gives an inner
product, so by applying the Gram--Schmidt process we can find an orthonormal basis $\{e_1, \ldots e_k\}$ for $r$ restricted to $V_k$. The inner product on $V_k$ gives a duality isomorphism from $V_k$ to $V^*_k$. Let $v_p$ denote the vector in $V_k$ that is dual to the functional $p\restriction_{V_k}$ via this isomorphism. By applying an isometry of the Euclidean space $V_k$ if necessary, we may assume that $v_p$ is a non-negative multiple of $e_1$. When one writes $r$, $c$ and $p$ with respect to the basis
$\{e_1, \ldots e_n\}$ we see from \Cref{lemma:directSumDecomposition} that they take the desired form.

Given that the market is of this form,
$i$ can be invariantly defined as the expected payoff of a riskless portfolio of cost $1$. In the same circumstances, $g$ can be invariantly defined as the maximum value of $p$ among costless portfolios $v$ with $r(v,v) \leq 1$. It follows that $i$ and $g$ are uniquely determined.
\item We suppose that all risk-free portfolios have cost zero. 
Take $k=n-\dim V^0$. Let $\{e_{k+1},\ldots, e_{n}\}$ be a basis for $V^0$.
Extend this to get a basis $\{v_1,\ldots,v_k,e_{k+1},\ldots,e_{n}\}$ for $V$. 
Let $V_k$ denote the span of the $v_k$. It is an inner product space
with respect to $r$, so by applying the Gram-Schmidt process we can obtain
a basis $\{e_1,\ldots,e_k,e_{k+1},\ldots,e_{n}\}$ for $V$ with the $\{e_1, \ldots e_k\}$ orthonormal. By applying an isometry of $V_k$ if necessary, we may assume that the vector dual to $c$ via the inner product on $V_k$ is a positive multiple of $e_1$. By applying a further isometry of the space spanned
by $e_2, \ldots, e_k$, we may assume that the vector dual to $p$ via the inner product on $V_k$ lies in the span of $e_1$ and $e_2$.
Writing the market with respect to this basis now puts it into the desired form.

Given that the market is of this form,
$m$ can be defined invariantly as $1$ over the maximum cost of any portfolio $v$ with $r(v,v)=1$. Define $i^\prime$ invariantly as the payoff $p(v)$ of a portfolio with $r(v,v)=1$ that maximizes the cost. Now $i$ can be defined invariantly by $i^\prime = \frac{i}{m}$. $g$ can be defined invariantly as the maximum expected payoff of any costless portfolio $v$ with $r(v,v)=1$. 
\end{enumerate}
The proof for the case when $c=0$ is similar.
\end{proof}

To avoid considering financially-uninteresting special cases in
the sequel we make the following definition.
\begin{definition}
A Markowitz market is {\em non-degenerate} if:
\begin{enumerate}[(i)]
\item The market is arbitrage-free;
\item There are no valueless portfolios;
\item $c$ and $p$ are linearly independent.
\end{enumerate} 
\end{definition}

It follows from our theorem that all non-degenerate Markowitz
markets of dimension $n$ are of the form ${\cal M}_{n-1,0,g,i}$ or
${\cal M}_{n,m,g,i}$ with $m \in (0, \infty)$ and $g \in (0, \infty)$.

We have identified the set of non-degenerate Markowitz markets up
to isomorphism. We now ask what is the topology of this space?

For a fixed underlying vector space, $V$ we can choose an isomorphism to $\R^n$.  The space of bilinear forms on $V$ can then be viewed as a subspace of $\R^{n^2}$ and so can be given a topology. We can then
give the space of Markowitz markets on $V$ a topology. This topology doesn't depend upon the choice of isomorphism from $V$ to $\R^n$. Thus the space of Markowitz markets has a natural topology. The moduli space of Markowitz markets is defined to be the quotient of the space of Markowitz markets by the equivalence relation given by Markowitz isomorphisms.

With this terminology established we may now prove the following corollary
of \Cref{theorem:classification}.
\begin{corollary}
The moduli space of non-degenerate Markowitz markets of dimension $n\geq 3$
is homeomorphic to the manifold with boundary $[0,\infty)\times(0, \infty) \times \R$. In particular, the map $\tau$ given by $\tau(m,g,i)={\cal M}_{n-\delta_0(m),m,g,i}$ is a homeomorphism. Here
$\delta_0(m)$ is equal to $1$ if $m=0$ and equal to $0$ otherwise.
\end{corollary}
\begin{proof}
It follows from \Cref{theorem:classification} that $\tau$ is a bijection.
 
Define $\tilde{\tau}(m,g,i)$ to be the market given in matrix form by
\[
{\bf r} = \left( \begin{array}{cc}
m^2 & 0 \\
0 & I_{n-1}
\end{array}
\right),\quad
c = (1,0,0,\ldots,0),\quad
p = (i,g,0,\ldots,0,).
\]
$\tilde{\tau}$ is continuous. The market $\tilde{\tau}(m,g,i)$
is Markowitz isomorphic to $\tau(m,g,i)$. Therefore $\tau$ is continuous.

We can invariantly and continuously associate a non-degenerate bilinear
form $\hat{r}$ with a non-degenerate Markowitz market by defining
\[
\hat{r}(u,v)=r(u,v)+c(u)c(v).
\]
To any non-degenerate bilinear form on a finite dimensional vector space, there is an associated isomorphism between the vector space and its dual. This isomorphism is associated continuously. Thus
we can continuously and invariantly associate a bilinear form acting on $V^*$ with any non-degenerate Markowitz market. We will write $\hat{r}^*$ for this form.

A short calculation shows that in both cases (i) and (ii) of \Cref{theorem:classification} we have $\hat{r}^*(c,c)=\frac{1}{1+m^2}$. 
Therefore
\[
m = \sqrt{\frac{1}{\hat{r}^*(c,c)}-1}.
\]
Thus the function $m$ defined on the moduli space of non-degenerate markets is continuous. We calculate
similarly that $\hat{r}^*(p,c)=\frac{i}{1+m^2}$ and $\hat{r}^*(p,p)=\frac{i^2}{1+m^2} + g^2$. Thus $m$, $i$ and $g$ are continuous functions on the moduli space of non-degenerate Markowitz markets.
Hence $\tau^{-1}$ is continuous.
\end{proof}

\begin{remark}
\label{remark:linearSDE}
We have called our algebraic structure a Markowitz market to emphasize its
financial relevance. However, this same structure occurs naturally in the abstract setting of linear stochastic differential equations.
Let $X_t$ be a stochastic process in an $n$-dimensional vector space $U$
determined by a linear stochastic differential equation driven by
$n$-dimensional Brownian motion with initial condition given by a known value for $X_0$. In coordinates we may write:
\[
\ed X^i_t = (\mu)^i \ed t + \sum_{i=1}^n (\sigma)^{ij} W^j_t
\]
for constants $\mu^i$ and $\sigma^{ij}$.
We will say that two such processes $X^1_t \in U^1$ and $X^2_t \in U^2$ are equivalent if there exists an isomorphism of $T:U^1\to U^2$ such that
$T X^1_t=X^2_t$ in distribution. We may associate a Markowitz market
to an SDE by taking the vector space $V=U^*$ and defining forms $a$, $b$ and $r$ as follows:
\[
c(\alpha) = \alpha(X_0) \text{ for } \alpha \in U^*;
\]
\[
p(\alpha) = \alpha\left({\mathbb E}\left(\frac{X_t}{t}\right)\right) \text{ for } \alpha \in U^*;
\]
\[
r(\alpha,\beta) = [\alpha(X),\beta(X)]_t \text{ for } \alpha, \beta \in U^*
\]
where $[Y^1,Y^2]_t$ denotes the quadratic covariation of two processes $Y^1_t$ and $Y^2_t$. Note that the definitions of $b$ and $a$ are independent of the choice of $t>0$. As is clear from our coordinate free definitions for $a$, $b$ and $r$, these forms are defined independently of the choice of basis for $\R^n$.
It is easy to see that we have established a one-to-one correspondence between Markowitz markets and linear stochastic differential equations. Thus our theorems can be interpreted as giving a partial classification
of linear stochastic differential equations up to linear transformation. We say that this is a partial classification since in this more general context, the ``arbitrage-free'' assumption may no longer be very natural and one should consider additional cases. We do not explore this further in this paper as 
our focus is on financial applications.
\end{remark}

\section{Portfolio Optimization}
\label{section:optimization}

Armed with our classification theorem, the study of portfolio optimization in Markowitz markets becomes entirely trivial.
\begin{definition}
Given a Markowitz market, a portfolio $v_0$ is said to be {\em risk minimizing}
if its risk $r(v_0,v_0)$ is equal to the minimum risk among all portfolios, $v$, with $c(v)=c(v_0)$ and $p(v)=p(v_0)$.
\label{defn:riskMinimizing}
\end{definition}
\begin{theorem}[Two mutual-fund theorem]
In a non-degenerate Markowitz market with no risk-free portfolios, the set of risk-minimizing portfolios is a vector subspace of $V$ of dimension at most $2$.
Moreover, for any feasible payoff and cost there is an associated risk-minimizing portfolio. This is called the two mutual-fund theorem because
the space of risk-minimizing portfolios can spanned by two portfolios, these are the ``mutual-funds''.
\label{thm:twoMutualFund}
\end{theorem}
\begin{proof}
Since there are no non-zero risk-free portfolios, we are in case (ii) of our classification, \Cref{theorem:classification}. In this case, our
vector space is Euclidean space with risk measured by distance, making the result geometrically obvious. We give a few formal details for completeness.

Two portfolios $v$ and $v_0$ have the same cost
and expected payoff if and only if their first two components are equal. The risk is equal to the sum of the squares of the components, and hence is minimized by taking all components other than the first two equal to zero. Hence
the space of risk-minimizing portfolios is the vector space spanned by
the standard basis vectors $\{e_1, e_2\}$.
\end{proof}
\begin{theorem}[One mutual-fund theorem]
In a non-degenerate Markowitz market with a risk-free portfolio 
the set of risk-minimizing portfolios is a vector subspace of $V$ of dimension
at most $2$ and contains the risk-free portfolio. For any feasible payoff and cost there is an associated risk-minimizing portfolio. This is called the one mutual-fund theorem because
the space of risk-minimizing portfolios can spanned by one arbitrary portfolio 
and a risk-free portfolio.
\end{theorem}
\begin{proof}
An obvious consequence of case (i) of \Cref{theorem:classification}
\end{proof}

We have not yet used the concept of {\em return} of a portfolio. In standard
treatments of Markowitz's theory it is usual to rescale investment problems in terms of the initial cost of a portfolio. This rescaling function is non-linear and not even defined for portfolios of zero cost. It often seems to unnecessarily complicate the discussion. For example, we
have stated the mutual-fund theorems in terms of vector spaces which we believe makes them
much easier to understand than conventional presentations. 

However, the idea that one might be able to rescale and transform a market to simplify it is central to our discussion; it is simply that returns are the ``wrong'' rescaling. We have observed that the covariance structure $r$ defines a natural length scale for the problem and have
transformed our coordinates ao that this becomes the standard Euclidean metric. This transformation has the advantage of being linear. This observation is generally useful throughout probability theory: covariance matrices define natural length scales.

\begin{definition}
The {\em expected return} of a portfolio, $v$ with non-zero cost is given by
\begin{equation}
\ER(v):=\frac{p(v)-c(v)}{c(v)}.
\label{eq:er}
\end{equation}

The {\em relative risk} of such a portfolio is given by
\[
\RR(v):= \frac{\sqrt{r(v,v)}}{c(v)}.
\]
Let $\phi$ map the set $V \setminus (\ker c)$ to $\R^2$ by
$\phi(v)=(\RR(v), \ER(v))$. The image of $\phi$ is called
the {\em feasible set}. The image of the set of risk-minimizing portfolios
is called the {\em efficient frontier}. The shape of the efficient frontier was identified in \cite{merton}.
\end{definition}

\begin{theorem}
In a non-degenerate Markowitz market, ${\cal M}_{n-\delta_0(m),m,g,i}$ 
with $n\geq 2$, the efficient frontier consists of the points $(x,y)$ with $x\geq0$ and
\begin{equation}
g^2(x^2 - m^2) = (y+1-i)^2.
\label{eq:hyperbola}
\end{equation}
When $n=2$, the feasible set is equal to the efficient frontier. When $n >2$, 
the feasible set is the set of all points on, or to the right of, the efficient frontier.
\end{theorem}
\begin{proof}
We consider first case (ii) of \Cref{theorem:classification} when $m>0$. Because of the scaling by cost in the definition of $\ER$ and $\RR$ we see that we need only consider the image of portfolios of cost $1$.

An efficient
portfolio with cost $1$ takes the form $v=(m,\lambda,0,\ldots,0)$ for some $\lambda$. It is mapped to:
\[
\phi(v)=\left(\sqrt{m^2 + \lambda^2}, i + g \lambda - 1 \right).
\]

We can compute $g^2 \lambda^2$ from either the $x$-coordinate or $y$-coordinate of $\phi(v)$. Equating these expressions gives the expression \eqref{eq:hyperbola}.
Since $g\neq0$ we see that the $y$-coordinate of $\phi(v)$ can take any real value, so the efficient frontier is the right arm of the hyperbola satisfying \eqref{eq:hyperbola}.

If $n=2$ all portfolios are efficient. If $n>2$, the portfolio $(m,\lambda,\mu,0,\ldots,0)$ is mapped by $\phi$ to
\[
\left(\sqrt{m^2 + \lambda^2 + \mu^2}, i + g \lambda - 1 \right).
\]
So an
y point to the right of the efficient frontier is feasible.

The efficient frontier and feasible set are similarly easy to calculate in case (i) of \Cref{theorem:classification}.
\end{proof}

The feasible set and the efficient frontier are iconic images of Markowitz's theory. They are illustrated in \Cref{fig:efficientfrontier}.

\begin{figure}[htbp]
\centering
\includegraphics[width=0.4\linewidth]{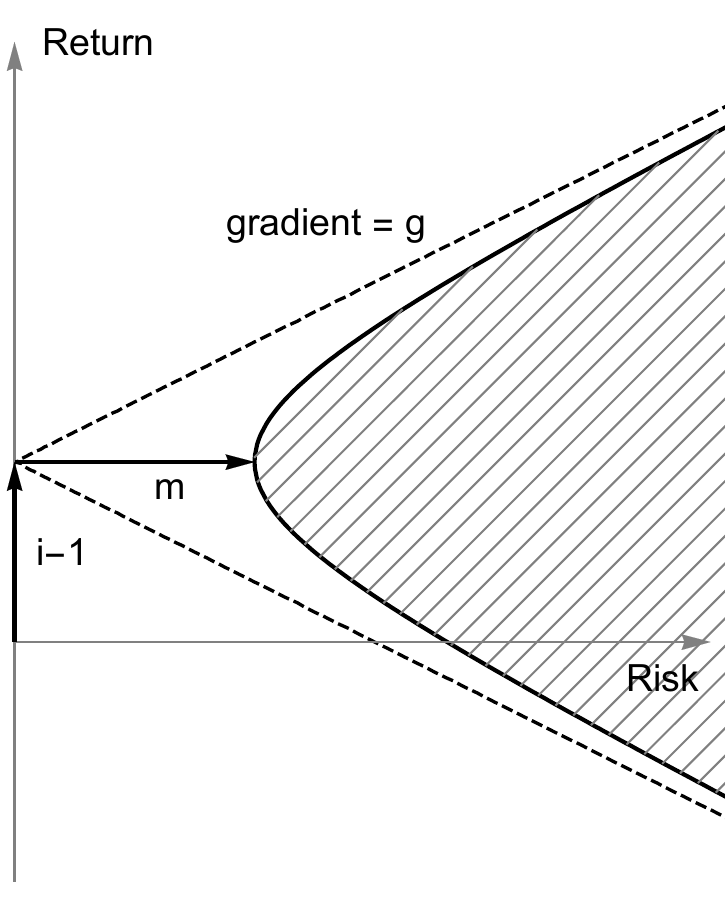}
\caption{The efficient frontier (curved line) and the feasible set (shaded).}
\label{fig:efficientfrontier}
\end{figure}

We now see the justification for our choice of parameter names for the space of Markowitz markets. The parameter $m$ measures the {\underline m}inimum risk of a portfolio of cost $1$,
the parameter $g$ measures the {\underline g}radient of the asymptotes when $m>0$ or the slope of the lines that the hyperbola degenerates to when. The parameter $i-1$ corresponds to the {\underline i}ntercept on the $y$-axis where the asymptotes meet.

From our point of view, the importance of the feasible set and the efficient frontier is explained by the following result.
\begin{theorem}
\label{theorem:pointedCategory}
Let $M_1$ and $M_2$ be two non-degenerate Markowitz markets of dimension $n$, Let $v_1$ and $v_2$ be portfolios in $M_1$ and $M_2$ repectively, each of cost 1. Then there exists a Markowitz isomorphism of $M_1$ and $M_2$ sending $v_1$ to $v_2$ if and only if the efficient frontiers of $M_1$ and $M_2$ are equal and $\phi(v_1)=\phi(v_2)$.
\end{theorem}
\begin{proof}
By assumption we are in either case (i) or case (ii) of \Cref{theorem:classification}.
We are in case (ii) if and only if the efficient frontier is one arm
of a  hyperbola.

In case (ii), our explicit formula for the efficient frontier shows that $m$, $g$ and $i$ can be recovered from its shape as shown in 
\Cref{fig:efficientfrontier}.

After a rotation of the inner product space spanned by $\{e_3,e_4,\ldots,e_k\}$, any portfolio in $M_1$ of cost $1$ can be written as $(m,\lambda,\mu,0,\ldots,0)$. The $\mu$ coefficient measures how far the image of $v$ under $\phi$ is to the right of the efficient frontier.
The $\lambda$ term identifies the point on the efficient frontier to the left of $\phi(v)$.

A similar argument can be applied in case (i).
\end{proof}

\Cref{theorem:pointedCategory} classifies the pointed category of non-degenerate markets with a marked portfolio of cost $1$. As we discussed in the introduction, this is the natural category to consider when comparing the peformance of a single portfolio to the market as a whole. 

As another example of how our approach can be generalized, consider the problem of comparing the performance of two portfolios within a market. The natural category is the category of Markowitz markets with two marked portfolios of cost $1$, which we will label $v^m$ and $v^i$. We think of $v^m$ as being a market portfolio, perhaps a stock index such as the S\&P 500, and $v^i$ being a specific portfolio whose performance we wish to evaluate. This situation can
be understood by the following classification theorem.

\begin{theorem}
\label{theorem:capm}
Let $M_1$ and $M_2$ be two non-degenerate Markowitz markets of dimension $n$
with marked portfolios $v^m_1$ and $v^i_2$ in $M_1$ and 
$v^m_2$ and $v^i_2$ in $M_2$ respectively. All the marked portfolios are of cost 1. We also assume that none of these portfolios are risk-free. There exists a Markowitz isomorphism of $M_1$ and $M_2$ sending $v^m_1$ to $v^m_2$ and $v^i_1$ to $v^i_2$ if and only if the efficient frontiers of $M_1$ and $M_2$ are equal, $\phi(v^m_1)=\phi(v^m_2)$, $\phi(v^i_1)=\phi(v^i_2)$ and $r(v^m_1,v^i_1)=r(v^m_2,v^i_2)$.
\end{theorem}
\begin{proof}
This is another geometrically obvious corollary of \Cref{theorem:classification}.
\end{proof}

Thus within any fixed market $M$ with a marked market portfolio $v^m$ the properties of a portfolio $v^i$ are determined entirely by $\phi(v^i)$ and the quantity $\beta^i:=\frac{r(v^m,v^i)}{r(v^m,v^m)}$. Thus this theorem gives a geometric interpretation of the Capital Asset Pricing Model and explains the central role of $\beta^i$ in this theory.

There is one feature of the market that is missed by risk-return diagrams, namely cost-free portfolios.
These portfolios are not uninteresting. In our case (ii) the costless portfolio $e_2$ provides one natural choice of mutual fund to use in the two mutual fund theorem. Adding multiples of this fund to your portfolio allows one to arbitrarily change the risk and return along the efficient frontier without affecting the cost. This fund is a particularly useful and easy to understand financial instrument.

Cost-free portfolios are also likely to be of great interest to rogue traders and fraudsters. They will want to know that arbitrarily large expected returns can be achieved in a Markowitz market at zero cost! Let us classify cost-free
portfolios for their benefit. We omit the proof.

\begin{theorem}
Define $\psi:V \to \R^2$ by $\psi(v)=(\sqrt{r(v,v)},p(v))$. The image of the cost-free, risk-minimizing portfolios under $\psi$ for the market ${\cal M}^n_{k,m,g,i}$ with $g>0$ is the set $(x,y) \in \R^2$ with $x\geq0$ and
\[
y = \pm g x.
\]
We call this set the {\em efficient frontier for costless portfolios}. 
The image of $\psi$ is either equal to the efficient frontier for costless portfolios or to the set of points on or to the right of the efficient frontier for costless portfolios.

There is an automorphism of the market mapping one costless portfolio to another if and only if they have the same image under $\psi$.
\end{theorem}

\begin{remark}
Let us see how our results can be applied beyond familiar portfolio optimization. In \Cref{remark:linearSDE} we noted that we have classified ``arbitrage-free''
linear stochastic differential equations up to weak equivalence. Thus two-mutual fund theorems should be expected when studying such equations. For example, consider the financial problem of optimizing expected utility when trading stocks that follow a multivariate Bachelier model (i.e.\ a linear stochastic differential equation). One sees from our invariance arguments that any meaningful solution to this problem will be a dynamic trading strategy in just two mutual funds. We say
any ``meaningful solution'' as it is not entirely straightforward to give
a mathematically rigorous formulation of this investment problem. Our point is that however this is done, invariance under Markowitz isomorphisms should be preserved, and this will result in some form of two mutual fund theorem.
\end{remark}

\section{Dimension reduction of Markowitz markets}
\label{section:dimensionReduction}

\red{Our clssification makes it easy to identify the interesting invariant subsets of the space of portfolios.}

\begin{theorem}
 For non-degenerate markets of dimension $n$ containing no-valueless portfolios, any invariant submanifold of the market under the automorphism group has dimension less than or equal to $2$ or greater than or equal to $n-2$. If $n>4$ then the invariant submanifolds of dimension
 less than or equal $2$ are all submanifolds of the set of risk-minimizing portfolios.
 \red{Furthermore in such markets, any invariant portfolio is an element of the set of risk-minimizing portfolios.}
\label{thm:dimensionReduction}
\label{lemma:invarianceTheorem}
\end{theorem} 
\begin{proof}
Any submanifold of the market which is closed under the automorphism group
of the market must consist of orbits of the automorphism group acting on $V$.
As we have seen, excluding costless portfolios, these orbits consist of the pre-image of points of $\phi$. The pre-image of the efficient frontier has dimension less than or equal to $2$. The pre-image of any other point in the feasible set is greater than or equal to $n-2$. We use the map $\psi$ to
apply similar reasoning to the case of costless portfolios. It follows that
invariant subspaces are of dimensions $0$, $1$, $2$, $n-2$, $n-1$ or $n$. If $n>4$, $n-2>2$. So in this case all low-dimensional invariant submanifolds are in the pre-image of the efficient frontiers. This implies they lie inside the set of risk-minimizing portfolios.

The final assertion is obvious.
\end{proof}

This result can be interepreted as a significant generalization of the classical two mutual fund
theorem. However, this interpretation of our result may seem obscure if the reader does not have a background
in areas of pure mathematics such as geometry where invariance arguments are commonplace. We will therefore explain this interpretation of our result from a theoretical point of view in Section \ref{section:invariantDefinitions}. We will then give a number of concrete financial applications: in Section \ref{section:uncertainty} we show how our theorem can be applied to the question of optimization under uncertainty;
in Section \ref{section:hedging} we show how our theorem can be applied to the question of choosing optimal hedging portfolios.

\subsection{Invariant definitions}
\label{section:invariantDefinitions}

There are two commonly used notions of invariance in mathematics. One such notion is {\em invariance under a group action}: if a group acts on a set one may ask which elements of the set are left unchanged by the group action. A second notion is {\em independence of presentation} where a mathematical property of an object only depends upon the isomorphism class of an object and not on any additional details used to describe the object. In this section we will formalize the latter notion in order to see how the two notions of invariance are related.

We begin by reviewing some fundamental definitions from category theory.

\begin{definition}
A category $C$ consists of the following data:
\begin{enumerate}[(i)]
	\item a class $\ob(C)$ of {\em objects}.
	\item a class $\hom(C)$ of {\em morphisms}. To each morphism $f$
		  are associated a source $a \in \ob(C)$ and target $b \in \ob(C)$. We write $f:a \to b$. $\hom(a,b)$ is the class of all morphisms from $a$ to $b$.
	\item for all $a, b, c \in \ob{C}$ a binary operation $\hom(a,b) \times \hom(b,c) \to \hom(a,c)$ called composition. If $f:a \to b$, $g:b \to c$ we write $g \circ f$ or just $g f$ for the composition.		  
\end{enumerate}
The composition satisfies
\begin{enumerate}[(i)]
\item Associativity: If  $f:a\to b$, $g:b\to c$, $h:c\to d$
\[	f \circ (g \circ h) = (f \circ g) \circ h \]
\item Identity: For all $x \in \ob(C)$ there exists a morphism ${\mathbf 1}_x:x \to x$
with the property that if $f:a \to x$, ${\mathbf 1}_x \circ f=f$ and if $g:x \to a$, $g \circ {\mathbf 1}_x = g$.
\end{enumerate}
\end{definition}

For example, we have already defined the category of Markowitz markets whose objects consist of quadruples of Markowitz markets and whose morphisms consist of Markowitz morphisms. The underlying set associated to each market is the set of vectors. We will call this category ${\cal M}$.

Note that in this case, and indeed the other cases that will interest us, the morphisms can be interpreted as functions and the composition law is given by ordinary function composition.

Another such category is $\Set$ the category
of all ``small'' sets. We must avoid talking about the set of all sets in order to avoid Russell's paradox. To resolve this problem one chooses a sufficiently large set that will contain all the sets of interest to you and define a small set to be sets contained in this large set. The same technical device can be applied to other categories, so we will henceforth allow ourselves to talk about ``all markets'' when we should say ``all small markets''.

\begin{definition}
A functor $F$ from a category $C$ to a category $D$ is a mapping which
\begin{enumerate}[(i)]
	\item associates to each object $x \in \ob(C)$ an object in $F(x) \in \ob(D)$.
	\item associates to a morphism $f:x \to y$ in $\hom(C)$ a morphism $F(f):F(x)\to F(y)$ in $\hom(D)$.
\end{enumerate}
and which satisfies
\begin{enumerate}[(i)]
	\item For all $x \in \ob(C)$, $F({\mathbf 1}_x) = {\mathbf 1}_{F(x)}$
	\item If $f:a \to b$ and $g:b \to c$ then $F(g \circ f)=F(g) \circ F(f)$.
\end{enumerate}
\end{definition}	

An obvious example of a functor is the identity map ${\cal M}\to{\cal M}$.

Another example is the ``forgetful functor'' which maps ${\cal M}$ to the category of sets by mapping an object $(V,r,c,p)\in \ob({\cal M})$ to the set of vectors in $V$. This functor acts as the identity on the morphisms of ${\cal M}$.

As a more interesting example of a functor, consider the category ${\cal V}^{\text{iso}}$ of finite dimensional vector spaces with
morphisms given by the invertible linear transformations between these vector spaces. We may then define a functor $F$ by $F(V)=V^*$, the dual of $V$ and $F(T)=(T^{-1})^*$ for a morphism $T:V \to W$. We will denote this functor by $(\cdot)^*$.

We have now established all the concepts we
need in order to define the notion of
an ``invariantly defined element''.

\begin{definition}
Let $C$ be a category and let $F$ be a functor from $C$ to $\Set$. Then an invariantly defined
element for $F$ is a map
\[
\phi: \ob(C) \to \Set
\]
such that $\phi(c) \in F(c)$ and $\phi(f c)=F(f)\phi(c)$ (recall that in set theory the elements of sets are themselves set which is why the codomain of $\phi$ is $\Set$ even
though we think of the values of $\phi$ primarily as elements rather than as sets).

If $F$ is a functor from category $C$ to category $D$ and if $D$ is a category whose morphisms are in fact transformations of a set, we will say that $\phi$ is an invariantly defined element for $F$ if it is an invariantly defined element for $U\circ F$ where $U$ is the forgetful functor.
\end{definition}

In particular if an invariantly defined element for the identity functor on ${\cal M}$ will be a map $\phi$ from a market to a portfolio in that market. So we will call this an invariantly defined portfolio. If we think of a morphism between two markets as a relabelling of the elements of the market, we see that an invariantly defined portfolio is a way of selecting a portfolio from any market that behaves correctly under relabellings. Thus our notion of an invariantly defined element captures the idea of ``independence of presentation''.

The advantage of our category theory approach is that we can define more than just invariantly defined portfolios. For example an invariantly defined element for the functor $(\cdot)^*$ will be called an invariantly defined linear functional. It is not hard to check that $p$ and $c$ are invariantly defined linear functionals.

We are now in a position to explain the
relationship between invariance under a group action and independence of presentation.

\begin{lemma}
[Invariance Lemma]
Let $C$ be a category where every morphism is invertible. Let $F$ be a functor from $C$ to $\Set$.

For each $c \in C$ write $\Aut c$ for the set of morphisms with source and target equal to $c$. $\Aut c$ forms a group under composition. It acts on the set $F(c)$ with the action defined by
\[
f(s)=F(f)(s).
\]
for $f \in \Aut c$ and $s \in F(c)$.

If $\phi$ is an invariantly defined element for  $F$ then $\phi(c)$ is invariant under $\Aut c$.

Conversely, let $C_c$ be the subcategory consisting of objects isomorphic to $c$ and their isomorphisms and let $s \in F(c)$ be invariant under $\Aut c$. The map given by:
\begin{equation}
\phi_{c,s}(c^\prime) = F(f)(s)
\label{eq:invarianceConverse}
\end{equation}
for any $f:c\to c^\prime$ is well-defined and
gives an invariantly defined element for $F|_{C_c}$
with $\phi_{c,s}(c)=s$.
\label{lemma:invarianceLemma}
\end{lemma}
\begin{proof}
The definition of a category ensures that $\Aut c$ is a semi-group. Our assumption that every morphism in $C$ is invertible ensures that $\Aut c$ is a group.

That the action given is a group action, follows from the definition of a functor. In detail if $f:c\to c$ and $g:c \to c$ then:
\[
(fg)(s)=F(fg)(s)=F(f)G(g)(s)=f(g(s))
\]
and
\[
F(\mathbf{1}_c)(s)=\mathbf{1}^{\Set}_{F(c)}(s)=s.
\]

If $\phi$ is an invariantly defined element of
$f$ and $f \in \Aut c$ then
\[
f \phi(c) = F(f) \phi(c)=\phi(f(c))=\phi(c).
\]
Here we have used in sequence the definition of the group action, the definition of an invariantly defined element and the fact that $f:c \to c$. Thus $\phi(c)$ is invariant under the action of $\Aut c$.

By definition of $C_c$, an isomorphism $f:c \to c^\prime$ exists for any $c \in C_c$. Suppose $g:c \to c^\prime$ too. Then $g^{-1} f \in \Aut c$. We see that
\[
F(f)(s)=F(g g^{-1} f)(s)=F(g)F(g^{-1}f)(s)=F(g)(s).
\]
The first equality is immediate, we then use the functorality of $F$ and then we use the invariance of $s$ under $\Aut c$. Thus the map $\phi_{c,s}$ defined by \eqref{eq:invarianceConverse} is well-defined as claimed. Suppose $f:c \to c^\prime$ then 
$f g : g c \to g c^\prime$ so
\[
\phi_{c,s}(g c^\prime)=F(gf)(s)=F(g)F(f)(s)=F(g)\phi_{c,s}.
\]
So $\phi_{c,s}$ is an invariantly defined element as claimed.

Finally note that
\[ \phi_{c,s}(c)=F({\mathbf 1}_c)(s)={\mathbf 1}^{\Set}_{F(c)} s = s \]
as claimed.
\end{proof}

A consequence of our invariance Lemma \ref{lemma:invarianceLemma} when combined with Theorem \ref{thm:dimensionReduction} is that
any invariantly defined portfolio must lie
in the given two dimensional space. Since we believe that any financially interesting statement must be independent of the labelling of stocks and mutual funds in a Markowitz market, this implies that the portfolios that can be identified uniquely by some financially interesting question all lie in a two dimensional space. This gives us our claimed generalization of the two mutual fund theorem, stated in the precise language of invariantly defined elements.

However, this does seem at first to open a new
problem, how can we tell if a given $\phi$ is invariantly defined? For example, if we fix constants $C$ and $P$, is $\phi_{C,P}$ given by
\begin{equation}
\phi_{C,P}((V,r,c,p)) = \underset{v \in V,\, c(v)=C,\, p(v)=P}{\argmin} r(v,v)
\label{eq:exampleInvariant}
\end{equation}
invariantly defined? We would certainly expect that it is, as this is surely a financially meaningful problem. But how can we prove this without a tedious calculation?

To resolve this problem we note that we can mirror most of the basic constructions of set theory using functors. We will restrict our attention to the case when every morphism in our category $C$ is invertible.

For example given two functors $F:C \to D_1$ and $G:C \to D_2$ we can define a product category $D_1 \times D_2$ in the obvious way.
This allows us to define the notion of an invariantly defined pair of elements.

Similarly if $F:C \to \Set$ is a functor, since $F(f)$ is permutation of $f(c)$ we may define an action of $F(f)$ on the power set ${\cal P}(f(c))$. Hence we can define a power-set functor ${\cal P}F$. This allows us to talk about invariantly defined sets of elements.

Since a function can be defined as a subset of a Cartesian product satisfying certain properties, we see that we can also talk about invariantly defined functions.

It is instructive to compute how we define a functor acting on functions in a little detail.
Let $F_S:C \to S$ and $F_T:C \to T$ be two functors to categories $S$ and $T$ which are backed by sets. Write $U$ for each of the forgetful functors to $\Set$. We wish to define a functor called $\fun(F_S,F_T)$ derived from $F_S$ and $F_T$. It will act on objects $c \in \ob{C}$ by
\[
\fun(F_S,F_T)(c)=\{ \psi:US(c)\to UT(c) \}
\]
Given $\psi:US(c)\to UT(c)$, we can view $\psi$ as a function in which case we write $\psi(x)$ in the usual way. We may also view $\psi$ as a set in which case we have $(x,\psi(x)) \in \psi$.
The recipe above tells us how we should define
the action of $\fun(F_S,F_T)$ on morphisms $f:c\to c^\prime$. The quantity $(\fun(F_S,F_T)f)\psi$ should be a new function which we can write explicitly as a subset of the Cartesian product:
\begin{align*}
(\fun(F_S,F_T)f)\psi&=\{ (F(f)s,F(f)t):(s,t) \in \psi \}
\end{align*}
We now translate this definition into conventional function notation.
\begin{align*}
z &= (\fun(F_S,F_T)(f)\psi)(s) \\
\iff \quad 
(s,z) &\in \{ (F_S(f)s,F_T(f)t):(s,t) \in \psi \} \\
\iff \quad 
(s,z) &\in F_T(f) \psi F_S(f)^{-1} \\
\iff \quad
z &= (F_T(f) \psi F_S(f^{-1}))(s)
\end{align*}
So in conventional function notation
\begin{align*}
(\fun(F_S,F_T)(f)\psi)(s) = (F_T(f) \psi F_S(f^{-1}))(s).
\end{align*}

Note that our definition of the dual space functor $(\cdot)^*$ which we defined earlier is simply a special case. Let us write  ${\mathbf 1}$ for the identity functor on vector spaces. Let us write $F_\R$ for the trivial functor which maps all vector spaces to $\R$ and all morphisms to the identity. We see that $(\cdot)^*=\fun({\mathbf 1}, \R)$.

In summary, we have shown that our definition of an invariantly defined element encompasses
many of the basic notions of set theory. In particular we have shown how the notion of an invariantly defined function follows directly from the set-theoretic definition of a function.

It is easy to check that all the properties one might expect of invariantly defined sets hold. For example, the union, intersection, product and power set of invariantly defined sets are all invariantly defined. It follows from such basic set theoretic facts as this and the definition of a function as a set that the composition of invariant functions is invariant, the image of an invariant set by an invariant function is invariant and so forth.

There is one set theoretic construction, however, that is not necessarily invariantly defined. This is the act of making a choice. For example, if we simply choose a portfolio in every Markowitz market there is no reason to expect this to be invariantly defined.

We conclude that any mathematical operation applied to invariantly defined inputs will
result in an invariantly defined output unless that operation involves making an arbitrary choice. This is a consequence of the fact that mathematics can be modelled using set theory.

As a concrete example, we see that $\phi_{C,P}$ defined in
\eqref{eq:exampleInvariant} is an invariantly defined portfolio as claimed. This is a consequence of the fact that all the inputs are invariantly defined. For example we have already remarked that $c$ and $p$ are invariantly defined. Indeed this is an immediate consequence of \ref{lemma:invarianceLemma}, as is the fact that $r$ is invariantly defined. The ordering $<$ defined on $\R$ that is used by $\argmin$ is also invariantly defined simply because the functor we are using to $\R$ is trivial. For the same reason $C$ and $P$ are invariantly defined.

In short, very often quantities are manifestly invariantly defined because their definition does not involve choices.

Having said that, sometimes a quantity is invariantly defined without it being immediately obvious.

For example, consider the measure $\mu_r$
on a Markowitz market $(V,r,c,p)$ defined as follows: First choose an $r$-orthonormal basis and hence define an inner product space isomorphism from $\psi:\R^n\to V$; define the measure of a subset of $V$ to be the Lebesgue measure of $\psi^{-1}(V)$. This definition apparently depends upon the choice of the orthonormal basis and so is not manifestly invariantly defined. However, the determinant of an orthogonal transformation is always $\pm 1$ and so we see that this measure is in fact defined independently of the choice of basis. We have called the measure $\mu_r$ as it only depends upon $r$.

Once we have established that a quantity is invariantly defined, we may use it to define
other invariantly defined quantities. For example we may define the standard Gaussian measure on $(V,r,c,p)$ by
\[
\frac{1}{(2 \pi)^\frac{n}{2}} e^{-\frac{1}{2}r(v,v)} \mu_r
\]
here $n$ is the dimension of the vector space $V$ which is invariantly defined by undergraduate linear algebra. We conclude that the standard Gaussian measure is invariantly defined.
We will use the measure $\mu_r$ and the standard Gaussian measure to define other more complex
invariant objects in Section \ref{section:uncertainty} below.

\begin{remark}
If the reader is already familiar with category theory, they may wonder whether invariantly defined elements can be interpreted as natural transformations (see \cite{eilenbergmaclane} for a definition of a natural transformations). To see how this can be done, let $\phi$ be an invariantly defined element for a functor $F:C\to \Set$. Let $Z$
be the functor mapping every object in $C$ to $\{0\}$ and every morphism in $C$ to the identity. For each $c \in \ob(C)$, define a function $\eta_\phi(c):\{0\}\to \Set$
by $\eta_\phi(c)(0)=\phi(c)$. Then $\eta_\phi$ is a natural transformation from $Z$ to $F$.
\end{remark}

\subsection{Optimization under uncertainty}
\label{section:uncertainty}

We will now show how the the theory of Section \ref{section:invariantDefinitions} can be applied to give a concrete financial application of Theorem \ref{thm:dimensionReduction}.

It has been observed that the portfolios identified by Markowitz's theory are often badly behaved in practice. For example in \cite{blackLitterman}, Black and Litterman observe that these ``almost always ordain large short positions in many assets'' and they cite \cite{greenHollifield} and \cite{bestGrauer} as academic references on the types of problems that are experienced.

One source of these problems
with Markowitz's theory is the difficulty of estimating expected returns.  One approach
to selecting the expected return vector is to use
expert knowledge, but in the absence of this
specialist knowledge one might estimate expected returns using historical returns. We will refer to the Markowitz market obtained from the historic mean and covariance combined with current prices as the historic Markowitz model.
However, as discussed in the references above, it has been found that the historic Markowitz model performs poorly in practice.

One tempting approach to resolving this problem is to consider model uncertainty. Any statistical measure of the historic returns will have some uncertainty and this should be incorporated into the optimization problem. Both the expected returns and the covariance matrix of returns will be difficult to estimate from historic data. There are many approaches to optimization under uncertainty
and many of these have been applied to this investment problem.
For example, in \cite{jorion} a Bayesian
approach is used, in \cite{ceriaStubbs} a robust optimization approach is followed, and \cite{garlappiUppalWang}
uses an approach based on the multi-prior model of decision making.

For all of these approaches one must make some
additional modelling decisions, but in each case
there is a natural choice of how to do this based on the data of the historic Markowitz model. 

Let us give a concrete example. In a robust optimization approach one needs to choose a set, ${\cal P}$, of possible probability distributions for asset payoffs. One might decide to choose as ${\cal P}$ the set of Gaussian distributions which are within a certain Hellinger distance, $d$ of the standard Gaussian measure arising in the historic Markowitz model
(see \cite{ayJost} for a definition of the Hellinger metric). One can think of the Hellinger distance as a measure of the statistical dissimilarity of two distributions.

Let us write $r_\P$ for covariance form defined by a probability distribution $\P \in {\cal P}$ and $p_\P$ for the expected payoff associated with $\P$. A typical robust optimization problem would be to find
\begin{equation}
\underset{v, c(v)=C}{\argmin} \left( \underset{\P \in {\cal P}}{\max} (r_\P(v,v) - \lambda p_\P(v)) \right)
\label{eq:robustProblem}
\end{equation}
for a chosen value of a risk-aversion parameter $\lambda$ and portfolio cost $C$.

Despite the complexity of this set-up, we see that the problem is invariant under Markowitz isomorphisms. The key step is to note that ${\cal P}$ is invariantly defined. To see this first recall that the Hellinger metric is invariantly defined on the space of measures on a finite dimensional real vector space, even if we forget the extra structure of $r$, $c$ and $p$. The Gaussian measure is invariantly defined. The set of measures which are Gaussian can be invariantly defined using only the reference measure $\mu_r$. Thus ${\cal P}$ is invariantly defined and hence the set defined by \eqref{eq:robustProblem} is also invariantly defined.

Hence if this problem does have a unique solution, that solution must be a weighted sum of the portfolios identified by the two mutual fund theorem.

We need not restrict ourselves to using the Hellinger metric to find invariantly defined sets like ${\cal P}$. There are many metrics and divergences defined on the space of distributions such as the $L^p$ metrics, the Wasserstein metric and the Kullback--Leibler divergence (again see \cite{ayJost} for the necessary definitions). All of these are invariantly defined using only the structure $\mu_r$. Thus we may repeat our analysis using any of these methods of defining ${\cal P}$ and we will obtain the same result.

Similarly, as predicted by our theory the invariant multi-prior problem described in Section 2.2.3 of \cite{garlappiUppalWang}
and the invariant Bayesian problem described in Section 3.3 of \cite{garlappiUppalWang} also
identify linear combinations of the portfolios coming from the two mutual fund theorem.

These examples illustrate the general principle implied by Theorem \ref{thm:dimensionReduction}
that the observed problems with the historic Markowitz model cannot be fixed by simply using more advanced optimization concepts. One also requires extra data. Indeed the approaches of \cite{blackLitterman}, \cite{garlappiUppalWang}, \cite{jorion} and \cite{ceriaStubbs} all suggest
additional data that could be incorporated into
the optimization problem in order to identify
alternative portfolios.

\subsection{Optimal hedging}
\label{section:hedging}

We give a second financial application of Theorem \ref{thm:dimensionReduction}.

Suppose that a fund manager has already created two 
investment funds according to the two mutual fund theorem targeting investors who currently have no liabilities. However, an investor with existing liabilities will have different risk preferences as they may be able to take advantage of hedging opportunities in the market. To attract such investors, the fund manager wishes to create one additional fund which can be used for hedging. Due to the overheads of fund management, the fund manager only wishes to create one additional fund. They ask what would be the optimal choice of hedging fund?

In lieu of any data on the existing liabilities of potential investors, they assume that the potential investors have been investing in the stock market previously to build up their liability and were using an optimal investment strategy based on their own estimates the payoff functional $p$. Thus they speculate that the potential investors will have liabilities that are normally distributed around one of the risk-minimizing portfolios found in Theorem \ref{thm:twoMutualFund} with covariance given by the bilinear form $r$. As we saw in the previous example, there are many notions of optimality one could now use to define an optimal hedging fund. However, as before any reasonable definition of 
an optimal hedging fund will be invariant under Markowitz morphisms.

Without loss of generality, the fund manager can also ensure that the fund is independent of their existing funds and is scaled so that purchasing $1$ unit of the fund has a cost of $1$.  All funds satisfying these last two properties are isomorphic under Markowitz morphisms. Hence whatever notion of optimality the fund manager decides to employ, if it is Markowitz invariant it will fail to identify any optimal hedging fund. Hence it is impossible to identify such an optimal hedging fund without supplying more data.

\section*{Acknowledgements}

\noindent This paper emerged from discussions with Teemu Pennanen and Matthew Glover. It has also benefited from the comments and suggestions of Damiano Brigo and Umut Cetin, and from discussions with James Newton and Ashwin Iyengar.

\appendix

\section{Relationship with the matrix formulation of portfolio optimization}
\label{appendix}

For the reader's convenience we describe
in detail how to translate between standard
presentations of Markowitz's theory and our
account. We will use boldface to indicate vectors in $\R^n$ and the standard font weight to represent abstract vectors.

Associated to a vector $v \in V$ we have its concrete realisation $\bs v \in \R^n$. The $i$-th component  $v_i$ of $\bs v$ indicates the quantity of asset $i$ that is held in the portfolio $v$. We will use boldface $\bs p$ and $\bs c$ for the row vectors defined by requiring $\bs p \, \bs v = p(v)$ and $\bs c \, \bs v = c(v)$ respectively. We similarly write boldface $\bs r$ for the symmetric matrix defined by requiring $r(u,v)=u^\top \bs r \, v$.

Let $\bs \Lambda$ be the diagonal matrix with $(i,i)$-th entry given by
\[
\bs \Lambda_{i,i} = \frac{1}{c_i}
\]
where $c_i$ is $i$-th component of $\bs c$.
We define the porfolio weights, $\bs w$, of a portfolio of non-zero cost by
\[
\bs w(v) = \frac{1}{\bs c \, \bs v} \Lambda^{-1} \bs v.
\]

The sum of the components of $\bs w(v)$ is then always equal to $1$. We may write this condition as $\bs 1 \bs w=1$ where boldface $\bs 1$ is the row vector consisting of $n$ ones.

From \eqref{eq:er}, the expected return $\ER$ can be computed from $\bs w$ and satisfies
\begin{equation}
\ER(v)= (\bs p \bs \Lambda - \bs 1 )\bs w
= {\bs \mu^\top} \bs w
\label{eq:pAndMu}
\end{equation}
where $\bs \mu=(\bs p \bs \Lambda - \bs 1 )^\top$ is the vector whose $i$-th component is the expected return of asset $i$.
The relative risk $\RR$ similarly satisfies
\begin{equation}
\RR(v)=\sqrt{\bs w^T \bs \Lambda^\top \bs r \bs \Lambda \bs w}= \sqrt{{\bs w^T} \bs \Sigma {\bs w}}
\label{eq:rAndSigma}
\end{equation}
where $\bs \Sigma=\bs \Lambda^\top \bs r \bs \Lambda$ is the covariance matrix of returns.

Given a non-zero initial cost $C$ we can use the mapping $v \to \bs w$ to
translate between the classical Markowitz optimization problem
\[
\begin{array}{lcl}
\underset{\bs w \in \R^n}{\text{minimize}} & & {\bs w}^\top \bs \Sigma \bs w \\
\text{subject to} & & \bs \mu^T \bs w = R \\
\text{and} & & \bs 1 \bs w=1 \\
\end{array}
\]
and the problem 
\[
\begin{array}{lcl}
\underset{v \in V}{\text{minimize}} & & r(v,v) \\
\text{subject to} & & p(v) = (R+1)C \\
\text{and} & & c(v) = C \\
\end{array}
\]
which we solved in Theorem \ref{thm:twoMutualFund}. In the classical Markowitz problem
with no risk-free asset, one assumes that $\Sigma$, and hence $r$, is positive definite. This puts us in case (ii) of Theorem \ref{theorem:classification} with $k=n$.

To show how our approach compares to the classical
approach of Lagrange multipliers
we now give a numerical example of the computation
of the set of risk-minimizing portfolios
and the isomorphism class following a geometric approach. This
can be compared with \cite{zivot}
which performs similar calculations numerically using the Lagrange multiplier approach.

Following \cite{zivot}, we now suppose that we are given numeric values for the vector of expected returns on each asset and the associated covariance matrix as follows:
\[
\bs \mu = \left(
\begin{array}{c}
 0.0427 \\
 0.0115 \\
 0.0285 \\
\end{array}
\right), \quad 
\bs \Sigma = \left(
\begin{array}{ccc}
 0.0100 & 0.0018 & 0.0011 \\
 0.0018 & 0.0109 & 0.0026 \\
 0.0011 & 0.0026 & 0.0199 \\
\end{array}
\right).
\]
In the classical formulation of the Markowitz problem used in \cite{zivot}, the cost vector $\bs c$ is not specified. So we are free to assume that the price of the assets are scaled such that the price of one unit of the asset is equal to $1$. This implies
that $\bs c=\bs 1$ and so $\bs \Lambda$ is the identity matrix. Thus from \eqref{eq:pAndMu} we must take $\bs p = \bs \mu^\top + \bs 1$
and from \eqref{eq:rAndSigma}, $\bs r = \bs \Sigma$.

We begin by identifying the duals $\bs c^*$ and $\bs p^*$ of $c$ and $p$ with respect to $r$. The dual, $f^* \in V$ of a functional $f \in V^*$ with respect to $r$ is defined by the requirement
\[
f(v) = r(f^*,v) \quad \forall v.
\]
Hence if ${\bf f}$ denotes the row vector associated with $f$ we have
\[
\bs f \bs v = (\bs f^*)^\top \bs \Sigma \bs v.
\]
Hence the dual of $\bs f$ satisfies
\[
 \bs f^* = \bs \Sigma^{-1} \bs f^\top.
\]
For our concrete example we compute
that
\[
\bs c^*=\left(
\begin{array}{c}
 83.5148 \\
 69.2237 \\
 36.5906 \\
\end{array}
\right)\quad \bs p^*=\left(
\begin{array}{c}
 87.6374 \\
 69.3134 \\
 37.7831 \\
\end{array}
\right).
\]
We note that it follows immediately from our classification theorem that the space of risk-minimizing portfolios (see Definition \ref{defn:riskMinimizing}) is spanned by these two vectors. We can also identify the portfolio weights that minimize risk irrespective of the payoff. They are given
by
\[
{\bf w}=\frac{\bs c^*}{\bs c \bs c^*}=\left(
\begin{array}{c}
 0.441109 \\
 0.365626 \\
 0.193264 \\
\end{array}
\right).
\]
This matches the value obtained using Lagrange multipliers in \cite{zivot}.

Applying the Gram-Schmidt process to the basis $\{\bs c^*, \bs p^*,(1,0,0)^\top\}$ we obtain
the $\bs r$-orthonormal vectors
\[
\bs e_1 = 
\left(
\begin{array}{c}
 6.06953 \\
 5.0309 \\
 2.65926 \\
\end{array}
\right), \quad \bs e_2 =
\left(
\begin{array}{c}
 7.29732 \\
 -7.91836 \\
 0.621044 \\
\end{array}
\right), \quad \bs e_3 =
\left(
\begin{array}{c}
3.63619 \\ 3.03729 \\ -6.67348
\end{array}
\right).
\]
This completely determines an isomorphism of the form given in Theorem \ref{theorem:classification}.
In practice one would only apply the Gram-Schmidt process to the pair of vectors ${\bs c^*, \bs p^*}$ as that is sufficient to identify the vectors $\bs e_1$ and $\bs e_2$, and hence the isomorphism class of the market. To identify the isomorphism class we simply solve the equations
\begin{align*}
m {\bs c^*} &= \bs e_1 \\
{\bs p^*} &= \frac{i}{m} \bs e_1 + g \bs e_2.
\end{align*}
In this case we find
\[
m = 0.0727, \quad g=-0.2382, \quad i=1.0286.
\]
As one would expect these values match the ones that can be read off from the plot of the efficient frontier in Figure 1.3 of \cite{zivot} using our own Figure \ref{fig:efficientfrontier}.

\bibliography{markowitz}

\begin{thebibliography}{}

\bibitem [\protect \citeauthoryear {%
Ay%
, Jost%
, V{\^a}n~L{\^e}%
\BCBL {}\ \BBA {} Schwachh{\"o}fer%
}{%
Ay%
\ \protect \BOthers {.}}{%
{\protect \APACyear {2015}}%
}]{%
ayJost}
\APACinsertmetastar {%
ayJost}%
\begin{APACrefauthors}%
Ay, N.%
, Jost, J.%
, V{\^a}n~L{\^e}, H.%
\BCBL {}\ \BBA {} Schwachh{\"o}fer, L.%
\end{APACrefauthors}%
\unskip\
\newblock
\APACrefYearMonthDay{2015}{}{}.
\newblock
{\BBOQ}\APACrefatitle {Information geometry and sufficient statistics}
  {Information geometry and sufficient statistics}.{\BBCQ}
\newblock
\APACjournalVolNumPages{Probability Theory and Related
  Fields}{162}{1-2}{327--364}.
\PrintBackRefs{\CurrentBib}

\bibitem [\protect \citeauthoryear {%
Best%
\ \BBA {} Grauer%
}{%
Best%
\ \BBA {} Grauer%
}{%
{\protect \APACyear {1991}}%
}]{%
bestGrauer}
\APACinsertmetastar {%
bestGrauer}%
\begin{APACrefauthors}%
Best, M\BPBI J.%
\BCBT {}\ \BBA {} Grauer, R\BPBI R.%
\end{APACrefauthors}%
\unskip\
\newblock
\APACrefYearMonthDay{1991}{}{}.
\newblock
{\BBOQ}\APACrefatitle {On the sensitivity of mean-variance-efficient portfolios
  to changes in asset means: some analytical and computational results} {On the
  sensitivity of mean-variance-efficient portfolios to changes in asset means:
  some analytical and computational results}.{\BBCQ}
\newblock
\APACjournalVolNumPages{The review of financial studies}{4}{2}{315--342}.
\PrintBackRefs{\CurrentBib}

\bibitem [\protect \citeauthoryear {%
Black%
\ \BBA {} Litterman%
}{%
Black%
\ \BBA {} Litterman%
}{%
{\protect \APACyear {1992}}%
}]{%
blackLitterman}
\APACinsertmetastar {%
blackLitterman}%
\begin{APACrefauthors}%
Black, F.%
\BCBT {}\ \BBA {} Litterman, R.%
\end{APACrefauthors}%
\unskip\
\newblock
\APACrefYearMonthDay{1992}{}{}.
\newblock
{\BBOQ}\APACrefatitle {Global portfolio optimization} {Global portfolio
  optimization}.{\BBCQ}
\newblock
\APACjournalVolNumPages{Financial Analysts Journal}{48}{5}{28--43}.
\PrintBackRefs{\CurrentBib}

\bibitem [\protect \citeauthoryear {%
Ceria%
\ \BBA {} Stubbs%
}{%
Ceria%
\ \BBA {} Stubbs%
}{%
{\protect \APACyear {2006}}%
}]{%
ceriaStubbs}
\APACinsertmetastar {%
ceriaStubbs}%
\begin{APACrefauthors}%
Ceria, S.%
\BCBT {}\ \BBA {} Stubbs, R\BPBI A.%
\end{APACrefauthors}%
\unskip\
\newblock
\APACrefYearMonthDay{2006}{}{}.
\newblock
{\BBOQ}\APACrefatitle {Incorporating estimation errors into portfolio
  selection: Robust portfolio construction} {Incorporating estimation errors
  into portfolio selection: Robust portfolio construction}.{\BBCQ}
\newblock
\APACjournalVolNumPages{Journal of Asset Management}{7}{2}{109--127}.
\PrintBackRefs{\CurrentBib}

\bibitem [\protect \citeauthoryear {%
Eilenberg%
\ \BBA {} MacLane%
}{%
Eilenberg%
\ \BBA {} MacLane%
}{%
{\protect \APACyear {1945}}%
}]{%
eilenbergmaclane}
\APACinsertmetastar {%
eilenbergmaclane}%
\begin{APACrefauthors}%
Eilenberg, S.%
\BCBT {}\ \BBA {} MacLane, S.%
\end{APACrefauthors}%
\unskip\
\newblock
\APACrefYearMonthDay{1945}{}{}.
\newblock
{\BBOQ}\APACrefatitle {General theory of natural equivalences} {General theory
  of natural equivalences}.{\BBCQ}
\newblock
\APACjournalVolNumPages{Transactions of the American Mathematical
  Society}{58}{2}{231--294}.
\PrintBackRefs{\CurrentBib}

\bibitem [\protect \citeauthoryear {%
Garlappi%
, Uppal%
\BCBL {}\ \BBA {} Wang%
}{%
Garlappi%
\ \protect \BOthers {.}}{%
{\protect \APACyear {2006}}%
}]{%
garlappiUppalWang}
\APACinsertmetastar {%
garlappiUppalWang}%
\begin{APACrefauthors}%
Garlappi, L.%
, Uppal, R.%
\BCBL {}\ \BBA {} Wang, T.%
\end{APACrefauthors}%
\unskip\
\newblock
\APACrefYearMonthDay{2006}{}{}.
\newblock
{\BBOQ}\APACrefatitle {Portfolio selection with parameter and model
  uncertainty: A multi-prior approach} {Portfolio selection with parameter and
  model uncertainty: A multi-prior approach}.{\BBCQ}
\newblock
\APACjournalVolNumPages{The Review of Financial Studies}{20}{1}{41--81}.
\PrintBackRefs{\CurrentBib}

\bibitem [\protect \citeauthoryear {%
Green%
\ \BBA {} Hollifield%
}{%
Green%
\ \BBA {} Hollifield%
}{%
{\protect \APACyear {1992}}%
}]{%
greenHollifield}
\APACinsertmetastar {%
greenHollifield}%
\begin{APACrefauthors}%
Green, R\BPBI C.%
\BCBT {}\ \BBA {} Hollifield, B.%
\end{APACrefauthors}%
\unskip\
\newblock
\APACrefYearMonthDay{1992}{}{}.
\newblock
{\BBOQ}\APACrefatitle {When Will Mean-Variance Efficient Portfolios Be Well
  Diversified?} {When will mean-variance efficient portfolios be well
  diversified?}{\BBCQ}
\newblock
\APACjournalVolNumPages{The Journal of Finance}{47}{5}{1785--1809}.
\PrintBackRefs{\CurrentBib}

\bibitem [\protect \citeauthoryear {%
Jensen%
, Black%
\BCBL {}\ \BBA {} Scholes%
}{%
Jensen%
\ \protect \BOthers {.}}{%
{\protect \APACyear {1972}}%
}]{%
jensen}
\APACinsertmetastar {%
jensen}%
\begin{APACrefauthors}%
Jensen, M\BPBI C.%
, Black, F.%
\BCBL {}\ \BBA {} Scholes, M\BPBI S.%
\end{APACrefauthors}%
\unskip\
\newblock
\APACrefYearMonthDay{1972}{}{}.
\newblock
{\BBOQ}\APACrefatitle {The capital asset pricing model: Some empirical tests}
  {The capital asset pricing model: Some empirical tests}.{\BBCQ}
\newblock

\PrintBackRefs{\CurrentBib}

\bibitem [\protect \citeauthoryear {%
Jorion%
}{%
Jorion%
}{%
{\protect \APACyear {1986}}%
}]{%
jorion}
\APACinsertmetastar {%
jorion}%
\begin{APACrefauthors}%
Jorion, P.%
\end{APACrefauthors}%
\unskip\
\newblock
\APACrefYearMonthDay{1986}{}{}.
\newblock
{\BBOQ}\APACrefatitle {Bayes-Stein estimation for portfolio analysis}
  {Bayes-stein estimation for portfolio analysis}.{\BBCQ}
\newblock
\APACjournalVolNumPages{Journal of Financial and Quantitative
  Analysis}{21}{3}{279--292}.
\PrintBackRefs{\CurrentBib}

\bibitem [\protect \citeauthoryear {%
Lintner%
}{%
Lintner%
}{%
{\protect \APACyear {1965}}%
}]{%
lintner}
\APACinsertmetastar {%
lintner}%
\begin{APACrefauthors}%
Lintner, J.%
\end{APACrefauthors}%
\unskip\
\newblock
\APACrefYearMonthDay{1965}{}{}.
\newblock
{\BBOQ}\APACrefatitle {The valuation of risk assets and the selection of risky
  investments in stock portfolios and capital budgets} {The valuation of risk
  assets and the selection of risky investments in stock portfolios and capital
  budgets}.{\BBCQ}
\newblock
\APACjournalVolNumPages{The review of economics and statistics}{}{}{13--37}.
\PrintBackRefs{\CurrentBib}

\bibitem [\protect \citeauthoryear {%
Markowitz%
}{%
Markowitz%
}{%
{\protect \APACyear {1952}}%
}]{%
markowitz}
\APACinsertmetastar {%
markowitz}%
\begin{APACrefauthors}%
Markowitz, H.%
\end{APACrefauthors}%
\unskip\
\newblock
\APACrefYearMonthDay{1952}{}{}.
\newblock
{\BBOQ}\APACrefatitle {Portfolio selection} {Portfolio selection}.{\BBCQ}
\newblock
\APACjournalVolNumPages{The Journal of Finance}{7}{1}{77--91}.
\PrintBackRefs{\CurrentBib}

\bibitem [\protect \citeauthoryear {%
Merton%
}{%
Merton%
}{%
{\protect \APACyear {1972}}%
}]{%
merton}
\APACinsertmetastar {%
merton}%
\begin{APACrefauthors}%
Merton, R\BPBI C.%
\end{APACrefauthors}%
\unskip\
\newblock
\APACrefYearMonthDay{1972}{}{}.
\newblock
{\BBOQ}\APACrefatitle {An analytic derivation of the efficient portfolio
  frontier} {An analytic derivation of the efficient portfolio
  frontier}.{\BBCQ}
\newblock
\APACjournalVolNumPages{Journal of Financial and Quantitative
  Analysis}{7}{04}{1851--1872}.
\PrintBackRefs{\CurrentBib}

\bibitem [\protect \citeauthoryear {%
Mossin%
}{%
Mossin%
}{%
{\protect \APACyear {1966}}%
}]{%
mossin}
\APACinsertmetastar {%
mossin}%
\begin{APACrefauthors}%
Mossin, J.%
\end{APACrefauthors}%
\unskip\
\newblock
\APACrefYearMonthDay{1966}{}{}.
\newblock
{\BBOQ}\APACrefatitle {Equilibrium in a capital asset market} {Equilibrium in a
  capital asset market}.{\BBCQ}
\newblock
\APACjournalVolNumPages{Econometrica: Journal of the econometric
  society}{}{}{768--783}.
\PrintBackRefs{\CurrentBib}

\bibitem [\protect \citeauthoryear {%
Sharpe%
}{%
Sharpe%
}{%
{\protect \APACyear {1964}}%
}]{%
sharpeCAPM}
\APACinsertmetastar {%
sharpeCAPM}%
\begin{APACrefauthors}%
Sharpe, W\BPBI F.%
\end{APACrefauthors}%
\unskip\
\newblock
\APACrefYearMonthDay{1964}{}{}.
\newblock
{\BBOQ}\APACrefatitle {Capital asset prices: A theory of market equilibrium
  under conditions of risk} {Capital asset prices: A theory of market
  equilibrium under conditions of risk}.{\BBCQ}
\newblock
\APACjournalVolNumPages{The journal of finance}{19}{3}{425--442}.
\PrintBackRefs{\CurrentBib}

\bibitem [\protect \citeauthoryear {%
Sharpe%
\ \BBA {} Tint%
}{%
Sharpe%
\ \BBA {} Tint%
}{%
{\protect \APACyear {1990}}%
}]{%
sharpe}
\APACinsertmetastar {%
sharpe}%
\begin{APACrefauthors}%
Sharpe, W\BPBI F.%
\BCBT {}\ \BBA {} Tint, L\BPBI G.%
\end{APACrefauthors}%
\unskip\
\newblock
\APACrefYearMonthDay{1990}{}{}.
\newblock
{\BBOQ}\APACrefatitle {Liabilities-a new approach} {Liabilities-a new
  approach}.{\BBCQ}
\newblock
\APACjournalVolNumPages{The journal of portfolio management}{16}{2}{5--10}.
\PrintBackRefs{\CurrentBib}

\bibitem [\protect \citeauthoryear {%
Treynor%
}{%
Treynor%
}{%
{\protect \APACyear {1961}}%
}]{%
treynor}
\APACinsertmetastar {%
treynor}%
\begin{APACrefauthors}%
Treynor, J\BPBI L.%
\end{APACrefauthors}%
\unskip\
\newblock
\APACrefYearMonthDay{1961}{}{}.
\newblock
{\BBOQ}\APACrefatitle {Toward a theory of market value of risky assets} {Toward
  a theory of market value of risky assets}.{\BBCQ}
\newblock
\APACjournalVolNumPages{Unpublished manuscript}{6}{}{}.
\PrintBackRefs{\CurrentBib}

\bibitem [\protect \citeauthoryear {%
Zivot%
}{%
Zivot%
}{%
{\protect \APACyear {2013}}%
}]{%
zivot}
\APACinsertmetastar {%
zivot}%
\begin{APACrefauthors}%
Zivot, E.%
\end{APACrefauthors}%
\unskip\
\newblock
\APACrefYearMonthDay{2013}{August}{}.
\newblock
\APACrefbtitle {{P}ortfolio {T}heory with {M}atrix {A}lgebra. {In} the lecture
  notes ``{C}omputational Finance and Financial Econometrics''
  \url{https://faculty.washington.edu/ezivot/econ424/portfolioTheoryMatrix.pdf}.}
  {{P}ortfolio {T}heory with {M}atrix {A}lgebra. {In} the lecture notes
  ``{C}omputational finance and financial econometrics''
  \url{https://faculty.washington.edu/ezivot/econ424/portfolioTheoryMatrix.pdf}.}
\PrintBackRefs{\CurrentBib}

\end{thebibliography}
\bibliographystyle{apacite}

\end{document}